\documentclass[runningheads]{llncs}

\usepackage[T1]{fontenc}
\usepackage{amsmath, amssymb}
\usepackage{graphicx}
\usepackage{booktabs}
\usepackage{xcolor}
\usepackage{mathtools}
\usepackage{enumitem}
\usepackage{hyperref}
\usepackage{algorithm}
\usepackage{algpseudocode}
\hypersetup{colorlinks=true,linkcolor=blue,citecolor=blue,urlcolor=blue}

\newcommand{\Vrfy}{\ensuremath{\mathsf{Vrfy}}}             

\newcommand{\MSM}{\ensuremath{\mathsf{MSM}}}               

\title{\textsf{2G2T}: Constant-Size, Statistically Sound\\ MSM Outsourcing}

\author{Majid Khabbazian}
\authorrunning{}
\institute{University of Alberta, Canada}

\begin{document}
\maketitle

\begin{abstract}
    Multi-scalar multiplication (MSM), $\MSM(\vec{P},\vec{x})=\sum_{i=1}^n x_i P_i$, is a dominant computational kernel in discrete-logarithm–based cryptography and often becomes a bottleneck for verifiers and other resource-constrained clients. We present \textsf{2G2T}, a simple protocol for verifiably outsourcing MSM to an \emph{untrusted} server. \textsf{2G2T} is efficient for both parties: the server performs only two MSM computations and returns only two group elements to the client, namely the claimed result $A=\MSM(\vec{P},\vec{x})$ and an auxiliary group element $B$. Client-side verification consists of a single length-$n$ field inner product and only three group operations (two scalar multiplications and one group addition). In our Ristretto255 implementation, verification is up to $\sim 300\times$ faster than computing the MSM locally using a highly optimized MSM routine (for $n$ up to $2^{18}$). Moreover, \textsf{2G2T} enables \emph{latency-hiding verification}: nearly all verifier work can be performed while waiting for the server's response, so once $(A,B)$ arrives the verifier completes the check with only one scalar multiplication and one group addition (both independent of $n$). Finally, despite its simplicity and efficiency, we prove that \textsf{2G2T} achieves statistical soundness: for any (even unbounded) adversarial server, the probability of accepting an incorrect result is at most $1/q$ per query, and at most $e/q$ over $e$ adaptive executions, in a prime-order group of size $q$.
\end{abstract}

\section{Introduction}

Multi-scalar multiplication (MSM) is a ubiquitous computational kernel in prime-order groups and
lies at the heart of modern discrete-logarithm--based cryptography. Let $\mathbb{G}$ be a cyclic group of
prime order $q$, written additively, with scalar field $\mathbb{F}_q$. Given bases
$\vec{P}=(P_1,\dots,P_n)\in\mathbb{G}^n$ and scalars $\vec{x}=(x_1,\dots,x_n)\in\mathbb{F}_q^n$, the MSM is
\[
\MSM(\vec{P},\vec{x}) \;:=\; \sum_{i=1}^n x_i P_i \;\in\; \mathbb{G}.
\]
Even with state-of-the-art algorithms (e.g., Pippenger-style methods), MSM computation requires many group operations
and remains costly for large $n$ (see, e.g.,~\cite{Bernstein2011Ed25519}).

A key reason MSM is so prominent is that many protocols deliberately \emph{aggregate} many checks into a small number
of large group equations, so verification is dominated by one (or a small number of) large MSMs.
In our target applications, these MSMs are not only large, but also \emph{fixed-base} across many queries:
a system publishes (or deterministically derives) a long vector of public generators once, and then many subsequent
proofs/commitments are verified by taking different linear combinations of the \emph{same} bases.
Bulletproofs is a canonical example: verification of a range proof---and especially aggregated verification---boils
down to a large MSM in public generator vectors, and the protocol supports aggregation and batching to reduce
marginal verification cost~\cite{Bunz2018Bulletproofs}.

A second major source of large, fixed-base MSMs is vector commitments and inner-product--based arguments.
In IPA-based polynomial commitments as used in Halo/Halo2-style systems, the commitment key contains a vector
$\vec{G}\in\mathbb{G}^d$ of $d$ group elements, and committing evaluates an inner product
$\langle \vec{a},\vec{G}\rangle$, i.e., an MSM over $d$ bases~\cite{Halo2BookPCIPA,Bowe2019Halo}.
The Halo2 documentation makes explicit that $d=2^k$; consequently, even moderate parameters (e.g., $k=11$ so
$d=2048$) already induce MSMs with thousands of terms, while larger $k$ scales $d$ exponentially
\cite{Halo2BookPCIPA,Halo2BookDevTools}.
In deployments, such long generator vectors are typically instantiated in a reproducible ``nothing-up-my-sleeve''
manner---for example by hashing domain-separated strings to curve points using standardized hash-to-curve techniques
\cite{RFC9380}---and then reused across many verifications.

Beyond zero-knowledge, related fixed-base multi-exponentiation/MSM patterns also arise in distributed key generation
(DKG)~\cite{Pedersen1991ThresholdDKG,Gennaro1999DKG} and verifiable secret sharing (VSS)~\cite{Feldman1987VSS,Pedersen1992VSS}:
parties verify share-consistency relations against a fixed vector of public commitments (typically coefficient
commitments), and these checks reduce to multi-exponentiations/MSMs over those fixed bases.

\smallskip
\noindent
\textbf{Outsourcing motivation.}
Outsourcing becomes attractive when the MSM dimension $n$ is large enough that computing
$\MSM(\vec{P},\vec{x})$ locally dominates cost under latency, energy, or throughput constraints.
This arises in two common scenarios:
(i) \emph{resource-constrained clients/verifiers} (mobile/embedded devices, thin clients, or applications running under tight
energy/latency budgets) that must validate many proofs or commitments but cannot afford repeated large MSMs; and
(ii) \emph{high-throughput verifiers} that validate large volumes of statements, where MSM lies on the critical path and
must be aggressively parallelized to keep up with demand.
In both settings, the fastest MSM implementations rely on careful engineering and often benefit from parallel hardware
(e.g., wide SIMD and, in some deployments, GPUs/accelerators) that may be unavailable or impractical to deploy at the
verifier, motivating delegation to a more powerful \emph{server}.
Finally, in some deployments the natural \emph{server} is already present: the proof generator/prover already holds
the coefficient vector $\vec{x}$ and can compute the MSM on behalf of the verifier.

Delegation, however, immediately introduces an integrity problem: an untrusted server may be faulty or adversarial
and can return an arbitrary group element. Thus, MSM outsourcing falls squarely within the verifiable computation
setting: the client seeks strong correctness guarantees while expending substantially less work than locally computing the MSM.
This tension is formalized in foundational work on verifiable delegation and verifiable computing
\cite{Goldwasser2008GKR,Gennaro2010NIVC}. General-purpose verifiable computation is powerful but often too heavy to
deploy at the granularity of a ubiquitous cryptographic kernel. This motivates specialized protocols tailored to MSM
that provide strong integrity guarantees with minimal client-side verification overhead.

\smallskip
\noindent
\textbf{Our proposed method: \textsf{2G2T}.}
We present \textsf{2G2T}, a statistically sound protocol for verifiable MSM outsourcing with constant-size
server-to-client communication and fast client-side verification.\footnote{%
The mnemonic \textsf{2G2T} reflects (i) ``two group elements to transfer'': the server returns two group elements per query; and
(ii) that the protocol can seem ``too good to be true'' at first glance, as it combines constant-size communication,
$O(1)$ verifier group operations, and $1/q$-level statistical soundness.}
After a one-time setup for fixed bases $\vec{P}$, the server answers a query $\vec{x}$ with only two group elements:
$A$ claimed to equal $\MSM(\vec{P},\vec{x})$ and an auxiliary value $B=\MSM(\vec{T},\vec{x})$ for a precomputed public
merged-bases vector $\vec{T}$.
The client verifies using a single length-$n$ field inner product and $O(1)$ group operations (two scalar multiplications
and one addition), while the server performs exactly two MSMs (only one more than computing $A$ alone).

Our Rust implementation over Ristretto255 confirms that \textsf{2G2T} yields large practical gains in the large-$n$ regime:
verification is up to $\sim 300\times$ faster than computing a highly optimized MSM, and up to $\sim 3{,}000\times$ faster
than a na\"ive ``$n$ scalar multiplications'' baseline, for dimensions up to $n=2^{18}$.
Moreover, verification is \emph{latency-hiding}: essentially all verifier work (the inner product and one fixed-base scalar
multiplication) can be performed while waiting for the server's response, so once $(A,B)$ arrives the client performs only
one scalar multiplication and one group addition to complete the check.
We prove that \textsf{2G2T} is perfectly complete and statistically sound: for any (even computationally unbounded)
adversarial server, the probability that the verifier accepts an incorrect $A$ is at most $1/q$ per execution (roughly
$252$ bits for Ristretto255), and at most $e/q$ over $e$ adaptive executions
(Theorems~\ref{thm:2g2t-completeness}--\ref{thm:2g2t-soundness-multi}).

\smallskip
\noindent
\textbf{Encrypted MSM (EMSM) and composability.}
The closest prior work that directly targets MSM outsourcing against an \emph{untrusted} server while also hiding the
scalar vector $\vec{x}$ is the recent EMSM line of work~\cite{AbbaszadehHKM25}, motivated primarily by \emph{prover-side}
outsourcing where $\vec{x}$ is sensitive and must remain hidden from the server.
In contrast, we focus on the common verifier-side regime where $\vec{x}$ is public (or already known to the server) and the
bottleneck is simply evaluating large MSMs.
This includes, for example, SNARK/Bulletproofs-style \emph{verification}, where the MSM scalars are derived from the public
transcript and public inputs, and deployments where the prover already holds $\vec{x}$ and can act as the server.
In this public-scalar setting, \textsf{2G2T} provides an extremely lightweight integrity layer: the client verifies using a
single length-$n$ inner product plus $O(1)$ curve operations, and verification is naturally \emph{latency-hiding}.\footnote{%
As reported in~\cite{AbbaszadehHKM25}, EMSM's client time for large instances is dominated by a \emph{pre-query} encryption step
that must complete before the request can be sent, which leaves substantially less opportunity for latency-hiding.}

When scalar privacy \emph{is} required, \textsf{2G2T} composes cleanly with EMSM: EMSM can provide coefficient hiding, while
\textsf{2G2T} supplies the integrity check.
This avoids the additional EMSM execution used to harden against an untrusted server, reducing client-side work by up to a
factor of two, without increasing the server's MSM work or the number of group elements returned (relative to EMSM's untrusted
server setting).
At the same time, the integrity guarantee is upgraded to \emph{statistical} soundness, rather than relying on EMSM's
dual-LPN-style assumptions.

\section{Threat model and security goals}\label{sec:threat-model}

We consider the verifiable outsourcing setting for MSM.
A (potentially resource-constrained) client $\mathcal{C}$ holds bases $\vec{P}\in\mathbb{G}^n$ and, for each query,
coefficients $\vec{x}\in\mathbb{F}_q^n$, and wishes to obtain
$\MSM(\vec{P},\vec{x})$ with substantially less online work than computing the MSM locally.
The client interacts with an external server $\mathcal{S}$ that is powerful but \emph{untrusted}.

\paragraph{Adversarial server.}
The server may deviate arbitrarily from the prescribed algorithm: it may be faulty, malicious,
or economically motivated to return an incorrect result to save computation.
We model $\mathcal{S}$ as an adaptive adversary that, after observing the public setup information
and the full transcript of all prior executions (including whether the client accepted or rejected),
can choose future queries $\vec{x}$ and responses as an arbitrary function of its entire view.
We allow $\mathcal{S}$ to be computationally unbounded; thus our integrity guarantees are
information-theoretic/statistical rather than based on computational hardness.

\paragraph{Communication model.}
We assume the network provides authenticity of messages (or equivalently, that the client can bind
responses to the corresponding request and setup state); confidentiality of network traffic is not required
for soundness.

\paragraph{Information revealed.}
We do not aim to hide the scalars $\vec{x}$ from the server (this corresponds to the common verifier-side
use case where $\vec{x}$ is public, or already known to the party doing the computation).
Similarly, we do not require privacy for the MSM output beyond what is implied by the application.
The client keeps its protocol-specific secret state (e.g., the key material produced in the setup phase)
confidential.

\paragraph{Security goals.}
An MSM outsourcing protocol should satisfy the following standard properties.

\begin{itemize}
    \item \textbf{Completeness.} If $\mathcal{S}$ follows the protocol and returns correct values,
    then $\mathcal{C}$ accepts and outputs $A=\MSM(\vec{P},\vec{x})$ for every query $\vec{x}$.

    \item \textbf{Soundness (integrity).} No adversarial server $\mathcal{S}^\ast$ should be able to make
    $\mathcal{C}$ accept an incorrect MSM result except with small probability.
    In our setting, we target \emph{statistical soundness}: for any (even computationally unbounded) $\mathcal{S}^\ast$ and any
    adaptively chosen queries, the probability that the client accepts while outputting
    $A\neq \MSM(\vec{P},\vec{x})$ is negligible in the chosen soundness/security parameter
    (often written as $\le 2^{-\lambda}$ for $\lambda$-bit statistical soundness).
    In \textsf{2G2T}, this probability is at most $1/q$ per query (equivalently, about $\log_2(q)$ bits).

    \item \textbf{Efficiency.} The client's online work per query (including verification) should be
    significantly cheaper than locally computing $\MSM(\vec{P},\vec{x})$.
    We also seek small communication overhead.
\end{itemize}

\section{The \textsf{2G2T} Protocol}
\label{sec:2g2t}

At a high level, \textsf{2G2T} rests on two basic observations that become powerful when combined.
First, MSM is linear in the bases (under coordinate-wise operations): for all
$\vec{P},\vec{P}'\in\mathbb{G}^n$, $\alpha\in\mathbb{F}_q$, and $\vec{x}\in\mathbb{F}_q^n$,
\[
\MSM(\vec{P}+\vec{P}',\vec{x})=\MSM(\vec{P},\vec{x})+\MSM(\vec{P}',\vec{x})
\quad\text{and}\quad
\MSM(\alpha\vec{P},\vec{x})=\alpha\MSM(\vec{P},\vec{x}),
\]
where $\vec{P}+\vec{P}'$ and $\alpha\vec{P}$ denote coordinate-wise group addition and scalar multiplication.
Second, MSM becomes particularly cheap when all bases lie in the span of a single point $Q\in\mathbb{G}$
and the verifier knows the corresponding coefficients: for $\vec{Q}=(\rho_1Q,\dots,\rho_nQ)$ with known
$\vec{\rho}=(\rho_1,\dots,\rho_n)\in\mathbb{F}_q^n$,
\[
\MSM(\vec{Q},\vec{x})
\;=\;
\sum_{i=1}^n x_i(\rho_iQ)
\;=\;
\langle \vec{x},\vec{\rho}\rangle\,Q,
\]
which reduces to a field inner product and a single fixed-base scalar multiplication by $Q$.

\textsf{2G2T} leverages the above facts by having the client precompute a merged-bases vector
\[
\vec{T} \;=\; r\vec{P}+(\rho_iQ)_i \;=\; (rP_1+\rho_1Q,\dots,rP_n+\rho_nQ),
\]
for a secret scalar $r\in\mathbb{F}_q$ and secret vector $\vec{\rho}\in\mathbb{F}_q^n$, and provide $\vec{T}$ to the server
(which may cache it for subsequent queries).
On input a query vector $\vec{x}$, the server returns
\[
A \stackrel{?}{=} \MSM(\vec{P},\vec{x})
\qquad\text{and}\qquad
B \stackrel{?}{=} \MSM(\vec{T},\vec{x}).
\]
The client verifies using the relation $B \stackrel{?}{=} rA + \langle \vec{x},\vec{\rho}\rangle Q$.
Intuitively, acceptance forces the server's response to be consistent with a \emph{single} hidden scalar $r$ and a
verifier-known fixed-base component, yielding statistical soundness while keeping the verifier far below the cost of an MSM.
We give the full protocol description below, and summarize it in pseudocode in Algorithm~\ref{alg:2g2t}.

\paragraph{Precomputation (once).}
Given public bases $\vec{P}=(P_1,\dots,P_n)\in\mathbb{G}^n$ and a fixed public non-identity point
$Q\in\mathbb{G}\setminus\{0\}$, the client samples a secret scalar
$r\xleftarrow{\$}\mathbb{F}_q^{\*}$ and a secret random vector
$\vec{\rho}=(\rho_1,\dots,\rho_n)\xleftarrow{\$}\mathbb{F}_q^n$.
It defines the (implicit) trap bases
\[
\vec{Q} := (\rho_1Q,\dots,\rho_nQ)\in\mathbb{G}^n
\]
and computes the \emph{merged-bases} vector
\[
\vec{T} \;:=\; r\vec{P}+\vec{Q}
\;=\; (rP_1+\rho_1Q,\;\dots,\;rP_n+\rho_nQ)\in\mathbb{G}^n.
\]
The client uploads $\vec{T}$ to the server (and $\vec{P}$ if not already known), and may treat
$\vec{T}$ as a long-lived \emph{public key} while keeping $(r,\vec{\rho})$ as a long-lived
\emph{secret key}.

\paragraph{MSM outsourcing (per query).}
To outsource the computation of $\MSM(\vec{P},\vec{x})$, the client sends $\vec{x}$ to the server.
(If the server has not cached $\vec{P}$ and/or $\vec{T}$, the client sends them as well.)
The server returns two group elements
\[
A \stackrel{?}{=} \MSM(\vec{P},\vec{x})
\qquad\text{and}\qquad
B \stackrel{?}{=} \MSM(\vec{T},\vec{x}).
\]

\paragraph{Verification.}
The client computes the inner product
\[
s \;:=\; \langle \vec{x},\vec{\rho}\rangle \;=\; \sum_{i=1}^{n} x_i\rho_i \in \mathbb{F}_q
\]
and accepts iff
\[
B \;\stackrel{?}{=}\; rA + sQ.
\]

\begin{algorithm}[h]
\caption{\textsf{2G2T}: verifiable outsourcing for $\MSM(\vec{P},\vec{x})$}
\label{alg:2g2t}
\begin{algorithmic}[1]
\Statex \textbf{Input:} public $\mathbb{G}$ of prime order $q$, public $Q\in\mathbb{G}\setminus\{0\}$, bases $\vec{P}=(P_1,\dots,P_n)\in\mathbb{G}^n$
\Statex \textbf{Input (per query):} coefficients $\vec{x}=(x_1,\dots,x_n)\in\mathbb{F}_q^n$
\Statex \textbf{Output (per query):} $A\in\mathbb{G}$ or $\bot$

\Statex
\State \textbf{Precomputation / key setup (once).}
\State Sample $r \xleftarrow{\$}\mathbb{F}_q^{\*}$ and $\vec{\rho}=(\rho_1,\dots,\rho_n)\xleftarrow{\$}\mathbb{F}_q^n$
\For{$i\gets 1$ to $n$}
  \State $T_i \gets rP_i + \rho_i Q$
\EndFor
\State $\vec{T}\gets (T_1,\dots,T_n)$
\State Upload $\vec{T}$ to the server (and $\vec{P}$ if needed) and have the server cache them
\Comment{Client keeps $\mathsf{sk}:=(r,\vec{\rho})$ secret}

\Statex
\State \textbf{Online outsourcing (per query for $\vec{x}$).}
\State \textbf{Client $\to$ Server:} send $\vec{x}$ 
\State \textbf{Server:} compute $A\gets \MSM(\vec{P},\vec{x})$ and $B\gets \MSM(\vec{T},\vec{x})$
\State \textbf{Server $\to$ Client:} return $(A,B)$

\Statex
\State \textbf{Client verification (per query).}
\State Compute $s \gets \langle \vec{x},\vec{\rho}\rangle = \sum_{i=1}^n x_i\rho_i \in \mathbb{F}_q$
\If{$B \neq rA + sQ$}
  \State \Return $\bot$
\EndIf
\State \Return $A$
\end{algorithmic}
\end{algorithm}

\section{ \textsf{2G2T}: Completeness and Soundness}

\paragraph{Common setting.}
Throughout the remainder of this section, we work in the following setup.

\begin{definition}[\textsf{2G2T} setup]\label{def:2g2t-setup}
Let $\mathbb{G}$ be a cyclic group of prime order $q$ written additively, with scalar field $\mathbb{F}_q$.
Let $Q\in\mathbb{G}\setminus\{0\}$ be public and fix $\vec{P}=(P_1,\dots,P_n)\in\mathbb{G}^n$.
The \textsf{2G2T} precomputation samples
$r\xleftarrow{\$}\mathbb{F}_q$ and $\vec{\rho}=(\rho_1,\dots,\rho_n)\xleftarrow{\$}\mathbb{F}_q^n$,
and defines $\vec{T}\in\mathbb{G}^n$ by
\[
T_i := rP_i + \rho_i Q \qquad \text{for all } i\in[n].
\]
Given a query $\vec{x}\in\mathbb{F}_q^n$ and a server response $(A,B)\in\mathbb{G}^2$, the verifier
runs the check in Algorithm~\ref{alg:2g2t} and we denote its decision by
$\Vrfy(r,\vec{\rho};\vec{x},A,B)\in\{0,1\}$.
\end{definition}

\begin{theorem}[Perfect completeness of \textsf{2G2T}]\label{thm:2g2t-completeness}
In the setting of Definition~\ref{def:2g2t-setup}, for every $\vec{x}\in\mathbb{F}_q^n$, if the server returns
\[
A=\MSM(\vec{P},\vec{x})
\qquad\text{and}\qquad
B=\MSM(\vec{T},\vec{x}),
\]
then the verifier in Algorithm~\ref{alg:2g2t} accepts and outputs $A$.
\end{theorem}

\begin{proof}
Let $s:=\langle \vec{x},\vec{\rho}\rangle=\sum_{i=1}^n x_i\rho_i\in\mathbb{F}_q$.
Using $T_i=rP_i+\rho_iQ$ and expanding the MSM,
\[
B \;=\; \sum_{i=1}^n x_i T_i
\;=\; \sum_{i=1}^n x_i(rP_i+\rho_iQ)
\;=\; r\sum_{i=1}^n x_iP_i \;+\; \left(\sum_{i=1}^n x_i\rho_i\right)Q
\;=\; rA + sQ.
\]
Hence the verifier's check $B \stackrel{?}{=} rA+sQ$ passes and it outputs $A$.\qed
\end{proof}

Next, we present two lemmas that will be used in the soundness proofs (Theorems~\ref{thm:2g2t-soundness}
and~\ref{thm:2g2t-soundness-multi}).

\begin{lemma}[Independence of $r$ and $\vec{T}$]\label{lem:2g2t-indep}
In the setting of Definition~\ref{def:2g2t-setup}, the random variables $r$ and $\vec{T}$ are statistically independent.
Equivalently, for every $t\in\mathbb{G}^n$ and every $a\in\mathbb{F}_q$,
\[
\Pr[r=a \mid \vec{T}=t] \;=\; \frac{1}{q}.
\]
\end{lemma}
\begin{proof}
Fix any $a\in\mathbb{F}_q$ and condition on $r=a$. Then for each $i\in[n]$,
\[
T_i \;=\; aP_i + \rho_i Q,
\qquad \rho_i \xleftarrow{\$}\mathbb{F}_q \text{ independently.}
\]
Since $\mathbb{G}$ has prime order and $Q\neq 0$, the map $\rho\mapsto \rho Q$ is a bijection $\mathbb{F}_q\to\mathbb{G}$,
so $\rho_iQ$ is uniform in $\mathbb{G}$. Adding the fixed point $aP_i$ is a translation, hence $T_i$ is uniform in $\mathbb{G}$.
Independence of the $\rho_i$ implies that $\vec{T}$ is uniform in $\mathbb{G}^n$, i.e., for every $t\in\mathbb{G}^n$,
\[
\Pr[\vec{T}=t \mid r=a] \;=\; \frac{1}{q^n},
\]
which does not depend on $a$. Therefore $\Pr[\vec{T}=t \mid r=a]=\Pr[\vec{T}=t]$ for all $a,t$, thus $r$ and $\vec{T}$ are
independent. Since $r\xleftarrow{\$}\mathbb{F}_q$, we have $\Pr[r=a]=1/q$, and by independence
\[
\Pr[r=a \mid \vec{T}=t] \;=\; \Pr[r=a] \;=\; \frac{1}{q}.
\]
\qed
\end{proof}

\begin{lemma}[Uniqueness of an accepting scalar for an incorrect $A$]\label{lem:2g2t-unique}
Fix $\vec{P}\in\mathbb{G}^n$, $\vec{T}\in\mathbb{G}^n$, a query $\vec{x}\in\mathbb{F}_q^n$, and a pair $(A,B)\in\mathbb{G}^2$.
Let $A^\star:=\MSM(\vec{P},\vec{x})$ and for each $r\in\mathbb{F}_q$ define $\vec{Q}^{(r)}:=\vec{T}-r\vec{P}\in\mathbb{G}^n$.
If there exist two distinct scalars $r_1\neq r_2$ such that
\[
B \;=\; r_1A + \MSM(\vec{Q}^{(r_1)},\vec{x})
\qquad\text{and}\qquad
B \;=\; r_2A + \MSM(\vec{Q}^{(r_2)},\vec{x}),
\]
then $A=A^\star$. In particular, if $A\neq A^\star$, then there is \emph{at most one} scalar $r\in\mathbb{F}_q$
for which the verifier's check can accept $(A,B)$.
\end{lemma}

\begin{proof}
Subtract the two equalities:
\[
0 \;=\; (r_2-r_1)A + \MSM(\vec{Q}^{(r_2)}-\vec{Q}^{(r_1)},\vec{x}).
\]
Since $\vec{Q}^{(r_2)}-\vec{Q}^{(r_1)}=(\vec{T}-r_2\vec{P})-(\vec{T}-r_1\vec{P})=(r_1-r_2)\vec{P}$, linearity of MSM gives
\[
\MSM(\vec{Q}^{(r_2)}-\vec{Q}^{(r_1)},\vec{x}) \;=\; (r_1-r_2)\MSM(\vec{P},\vec{x}) \;=\; (r_1-r_2)A^\star.
\]
Hence
\[
0 \;=\; (r_2-r_1)A + (r_1-r_2)A^\star \;=\; (r_2-r_1)(A-A^\star).
\]
Because $r_2\neq r_1$, we have $r_2-r_1\neq 0$ in $\mathbb{F}_q$, thus $A=A^\star$.\qed
\end{proof}

\begin{theorem}[Statistical soundness of \textsf{2G2T}]\label{thm:2g2t-soundness}
In the setting of Definition~\ref{def:2g2t-setup}, for any (possibly computationally unbounded)
adversarial server $\mathcal{S}^{\*}$ that, on input $(\vec{P},\vec{T},\vec{x})$ for an arbitrary query
$\vec{x}\in\mathbb{F}_q^n$, outputs $(A,B)\in\mathbb{G}^2$, we have
\[
\Pr\Bigl[\,\Vrfy(r,\vec{\rho};\vec{x},A,B)=1 \ \wedge\  A\neq \MSM(\vec{P},\vec{x})\,\Bigr]
\;\le\; \frac{1}{q},
\]
where the probability is over the client's randomness in the setup and any internal randomness used by the (possibly randomized) server.
\end{theorem}

\begin{proof}
Fix an arbitrary query $\vec{x}\in\mathbb{F}_q^n$ and let $(A,B)\in\mathbb{G}^2$ be the output of $\mathcal{S}^{\*}$ on
input $(\vec{P},\vec{T},\vec{x})$.  Write $A^\star:=\MSM(\vec{P},\vec{x})$.

Fix an arbitrary $t\in\mathbb{G}^n$ and condition on the event $\vec{T}=t$ (and on $\mathcal{S}^{\*}$'s internal randomness,
so that $(A,B)$ is fixed). By Lemma~\ref{lem:2g2t-indep}, the secret scalar $r$ is uniform over $\mathbb{F}_q$ under this
conditioning.

For each $\hat r\in\mathbb{F}_q$, define the derived vector
\[
\vec{Q}^{(\hat r)} \;:=\; \vec{T}-\hat r\,\vec{P} \;=\; t-\hat r\,\vec{P}\in\mathbb{G}^n.
\]
In particular, for the \emph{actual} secret $r$ used in setup, the relation
$\vec{T}=r\vec{P}+(\rho_iQ)_i$ implies $\vec{Q}^{(r)}=(\rho_iQ)_i$, and hence
\[
\MSM(\vec{Q}^{(r)},\vec{x})
\;=\;
\MSM\bigl((\rho_iQ)_i,\vec{x}\bigr)
\;=\;
\langle \vec{x},\vec{\rho}\rangle Q.
\]
Therefore the verifier's check $B \stackrel{?}{=} rA+\langle \vec{x},\vec{\rho}\rangle Q$ is equivalent to
\[
B \;=\; rA \;+\; \MSM(\vec{Q}^{(r)},\vec{x}).
\]

If $A\neq A^\star$, then by Lemma~\ref{lem:2g2t-unique} the above equality can hold for \emph{at most one} value of
$r\in\mathbb{F}_q$. Since $r$ is uniform in $\mathbb{F}_q$ conditioned on $\vec{T}=t$, it follows that
\[
\Pr\Bigl[\Vrfy(r,\vec{\rho};\vec{x},A,B)=1 \ \wedge\ A\neq A^\star \ \Bigm|\ \vec{T}=t\Bigr] \;\le\; \frac{1}{q}.
\]
Since $t$ was arbitrary, the same bound holds without conditioning, proving the theorem. \qed
\end{proof}

\begin{theorem}[$e$-execution statistical soundness of \textsf{2G2T}]\label{thm:2g2t-soundness-multi}
In the setting of Definition~\ref{def:2g2t-setup}, let $\mathcal{S}^{\*}$ be any (possibly computationally unbounded)
adaptive adversary that is given $(\vec{P},\vec{T})$ and then interacts with the verifier for at most $e$ executions.
In execution $j\in[e]$, $\mathcal{S}^{\*}$ adaptively chooses a query $\vec{x}^{(j)}\in\mathbb{F}_q^n$ and outputs a pair
$(A^{(j)},B^{(j)})\in\mathbb{G}^2$, after which the verifier returns the accept/reject bit.
Then
\[
\Pr\Bigl[\exists\,j\in[e]\!:\ \Vrfy(r,\vec{\rho};\vec{x}^{(j)},A^{(j)},B^{(j)})=1
\ \wedge\ A^{(j)}\neq \MSM(\vec{P},\vec{x}^{(j)})\Bigr]
\;\le\; \frac{e}{q},
\]
where the probability is over the client's randomness in the precomputation and the internal randomness of $\mathcal{S}^{\*}$.
\end{theorem}
\begin{proof}
Fix any adaptive adversary $\mathcal{S}^{\*}$ and consider the interaction for at most $e$ executions.
For $j\in\{1,\dots,e\}$, let $\mathsf{Bad}_{\le j}$ be the event that an incorrect acceptance occurs in one of the first $j$
executions, and let $\mathsf{Bad}:=\mathsf{Bad}_{\le e}$.

Fix an arbitrary $t\in\mathbb{G}^n$ and condition on $\vec{T}=t$ (and on $\mathcal{S}^{\*}$'s internal randomness, so that
$\mathcal{S}^{\*}$ is deterministic given the verifier's replies). By Lemma~\ref{lem:2g2t-indep}, under this conditioning
the secret scalar $r$ is uniform over $\mathbb{F}_q$.

We define a sequence of candidate sets $R_j\subseteq\mathbb{F}_q$ of possible values of $r$ consistent with the transcript
up to execution $j$, \emph{conditioned on no incorrect acceptance so far}. Set $R_0:=\mathbb{F}_q$.

Consider execution $j\in[e]$ and condition on $\neg\mathsf{Bad}_{\le j-1}$.
Let $\vec{x}^{(j)}$ be the query chosen by $\mathcal{S}^{\*}$ and let $(A^{(j)},B^{(j)})$ be its response.
Write $A^{\star (j)}:=\MSM(\vec{P},\vec{x}^{(j)})$.

\smallskip
\noindent\emph{Case 1: $A^{(j)}=A^{\star (j)}$.}
Then the verifier's check
$B^{(j)} \stackrel{?}{=} rA^{(j)}+\langle \vec{x}^{(j)},\vec{\rho}\rangle Q$
is equivalent to
$B^{(j)} \stackrel{?}{=} \MSM(\vec{T},\vec{x}^{(j)})=\MSM(t,\vec{x}^{(j)})$,
which is independent of $r$ given $\vec{T}=t$.
In particular, this execution cannot be an \emph{incorrect} acceptance, and the verifier's bit does not shrink the set of
possible $r$ values. Thus we set $R_j:=R_{j-1}$.

\smallskip
\noindent\emph{Case 2: $A^{(j)}\neq A^{\star (j)}$.}
By Lemma~\ref{lem:2g2t-unique}, for the fixed tuple $(\vec{P},t,\vec{x}^{(j)},A^{(j)},B^{(j)})$ there exists \emph{at most one}
scalar $\hat r_j\in\mathbb{F}_q$ such that the verifier would accept in execution~$j$.
Therefore, conditioned on $\neg\mathsf{Bad}_{\le j-1}$ and $\vec{T}=t$, an incorrect acceptance in execution~$j$ can occur
only if $r=\hat r_j$; since $r$ is uniform over $R_{j-1}$ under this conditioning, we have
\[
\Pr[\text{execution $j$ is an incorrect acceptance}\mid \neg\mathsf{Bad}_{\le j-1},\,\vec{T}=t]
\;\le\;
\frac{1}{|R_{j-1}|}.
\]
Moreover, on the event $\neg\mathsf{Bad}_{\le j}$ the verifier must have rejected in execution $j$, which implies
$r\neq \hat r_j$ and hence removes at most one candidate value from $R_{j-1}$. Thus we set
$R_j := R_{j-1}\setminus\{\hat r_j\}$ (if $\hat r_j\in R_{j-1}$), so $|R_j|\ge |R_{j-1}|-1$.

\smallskip
\noindent\textbf{Bounding the total failure probability.}
By construction, $|R_0|=q$ and in each execution we remove at most one value of $r$ as long as no incorrect acceptance occurs.
Hence on the event $\neg\mathsf{Bad}_{\le j-1}$ we have $|R_{j-1}|\ge q-(j-1)$, and therefore
\[
\Pr[\text{execution $j$ is an incorrect acceptance}\mid \neg\mathsf{Bad}_{\le j-1},\,\vec{T}=t]
\;\le\;
\frac{1}{q-(j-1)}.
\]
Let $p_j:=\Pr[\mathsf{Bad}_{\le j}\mid \vec{T}=t]$. Then
\[
p_j \;\le\; p_{j-1} + (1-p_{j-1})\cdot \frac{1}{q-(j-1)}.
\]
Using $p_0=0$, a straightforward induction gives $p_j\le j/q$ for all $j\le e$, and in particular
$p_e\le e/q$.
Since $t$ was arbitrary, the same bound holds without conditioning on $\vec{T}$, proving the theorem. \qed
\end{proof}

\section{\textsf{2G2T}: Performance}
\label{sec:2g2t-performance}

This section summarizes the concrete performance costs of \textsf{2G2T} (Algorithm~\ref{alg:2g2t})
in terms of (i) communication, (ii) server computation, and (iii) client computation.
We emphasize the \emph{online} costs incurred per outsourced MSM query, and separately note the one-time
precomputation overhead.

\paragraph{Communication.}
\textsf{2G2T} adds only minimal communication beyond what is intrinsically required to outsource an MSM.
An untrusted server cannot compute $\MSM(\vec{P},\vec{x})$ without access to the bases $\vec{P}$ and the
scalars $\vec{x}$; thus, transmitting $\vec{P}$ (if not already known) and $\vec{x}$ (if not already known)
is unavoidable in the generic outsourcing setting.

\begin{itemize}
  \item \textbf{One-time (setup) communication.}
  The client uploads the merged-bases vector $\vec{T}\in\mathbb{G}^n$ to the server once, and the server may cache it.
  If the server does not already have the bases $\vec{P}$, the client also uploads $\vec{P}$ once.
  Thus the one-time upload is $\vec{T}$ (and optionally $\vec{P}$), i.e., $n$ (or $2n$) group elements.

  \item \textbf{Per-query (online) communication.}
  The server returns exactly two group elements $(A,B)\in\mathbb{G}^2$.
  The client-to-server payload is simply the coefficient vector $\vec{x}$ \emph{unless} the server already holds it.
  In particular, in architectures where the server is also the component that produced (or already possesses)
  the coefficients $\vec{x}$ (e.g., a proving service or an internal accelerator invoked by a prover/verifier),
  the online communication from client to server can be reduced to a lightweight request/identifier, while the
  server-to-client response remains the constant-size pair $(A,B)$.

  \item \textbf{Overhead vs.\ non-verifiable outsourcing.}
  Compared to naïvely outsourcing only $A=\MSM(\vec{P},\vec{x})$, \textsf{2G2T} adds (i) a one-time upload of $\vec{T}$ and
  (ii) one additional returned group element $B$ per query.
\end{itemize}

\paragraph{Server computation.}
For each query, the server computes two MSMs:
\[
A=\MSM(\vec{P},\vec{x})
\qquad\text{and}\qquad
B=\MSM(\vec{T},\vec{x}).
\]
Thus, online server work is essentially \emph{one additional} MSM relative to computing $A$ alone.
In practice, implementations based on windowing (e.g., Pippenger-style bucket methods) can reuse the
scalar decomposition of $\vec{x}$ across both MSMs, since the scalars are identical and only the bases differ.
Consequently, the incremental cost of the second MSM is typically less than a full second pass of all scalar
preprocessing.

\paragraph{Client computation.}
Per query, the client performs no large MSM in $\mathbb{G}$.
Instead, verification consists of:
\begin{itemize}
  \item one length-$n$ inner product in the scalar field,
  \[
  s=\langle \vec{x},\vec{\rho}\rangle=\sum_{i=1}^n x_i\rho_i \in \mathbb{F}_q,
  \]
  \item two scalar multiplications in $\mathbb{G}$ (compute $rA$ and $sQ$),
  \item one group addition (compute $rA+sQ$) and one equality check against $B$.
\end{itemize}
In particular, the number of \emph{group} operations performed by the client per query is constant, independent of $n$.
As $n$ grows, the dominant client cost is the field-arithmetic inner product, while the elliptic-curve/group
work remains essentially fixed.

\paragraph{One-time precomputation overhead.}
The precomputation phase computes $\vec{T}=(T_1,\dots,T_n)$ with $T_i=rP_i+\rho_iQ$ for all $i$.
This incurs $n$ scalar multiplications by a fixed scalar $r$ on the bases $\vec{P}$, $n$ scalar multiplications
on the fixed base $Q$, and $n$ group additions.
Because this phase is executed once and amortized over many outsourced queries, its cost is typically dominated
by the online savings when many MSMs are outsourced over the same bases $\vec{P}$.
Moreover, fixed-base techniques can accelerate the $Q$-multiplications, since $Q$ is constant.

\section{Related work}\label{sec:related-work}

\paragraph{Verifiable MSM in transparent, pairing-free groups.}
The closest related line of work that explicitly targets verifiable multi-exponentiation/MSM in transparent groups is due to
Pinkas~\cite{Pinkas2025VerifiableMSM}, which gives a \emph{non-interactive, publicly verifiable} statistical check with
soundness parameter $2^{-\lambda}$.
At a high level, Pinkas' protocol attains public verifiability by returning a proof consisting of $O(\log q)$
group elements and by performing verifier work that depends on the target soundness parameter $\lambda$.
In contrast, \textsf{2G2T} is a designated-verifier protocol with a one-time keyed setup whose \emph{per-query} verification is
dominated by a single length-$n$ field inner product plus $O(1)$ group operations, while achieving constant-size
server-to-client communication (two group elements) and per-query soundness error $1/q$ in prime-order groups.
We view these as complementary points in the design space: public verifiability with larger proofs versus designated-verifier
verification with constant-size responses and ``full-group'' statistical soundness.

A second line of work considers MSM outsourcing for zero-knowledge proof generation.
In particular, Xiong \emph{et al.}~\cite{FCConsistencyOnlyMSM} propose an ``immediate verification'' step based on a
consistency relation between multiple values returned by the server.
However, such a consistency check does not by itself force the server to incorporate \emph{all} $n$ terms of the
MSM instance: a malicious server can satisfy the stated relation while computing only a strict subset of indices
(e.g., returning a partial sum such as $A=x_1P_1$), even though $A\neq \MSM(\vec{P},\vec{x})$ in general.
Thus, in the standard adversarial-server model for outsourcing, the described immediate-check mechanism does not, on its own,
imply soundness.

\paragraph{Structured (non-transparent) approaches via SRS and pairings.}
In bilinear (pairing-based) settings with a structured reference string (SRS), many large algebraic relations can be certified
succinctly via a constant number of pairing equations.
A canonical example is the Kate--Zaverucha--Goldberg (KZG) polynomial commitment scheme: the commitment is an MSM against
SRS-derived bases, yet a constant-size evaluation proof can be verified using $O(1)$ pairings, independent of the polynomial
degree~\cite{KZG10}.
This paradigm underlies many non-transparent SNARK constructions with very fast verification (e.g., Groth16 and
polynomial-commitment--based families), where the proof system itself provides an integrity layer that certifies the relevant
algebraic relations without requiring the verifier to redo large MSMs~\cite{Groth16,Sonic19,Marlin20}.
Our work targets a different regime: transparent, pairing-free prime-order groups and \emph{generic-base} MSM instances, where
the bases are not SRS-derived and there is no pairing-based shortcut for succinct certification.

\paragraph{Homomorphic authenticators and authentication-based delegation.}
Our check is closely related in spirit to \emph{homomorphic authenticators} (a.k.a.\ homomorphic MACs/signatures) for linear
functions: a client authenticates a dataset once, an untrusted server computes on authenticated data, and the verifier checks
that the output is the correct function evaluation.
Linearly homomorphic MACs were introduced in the context of network coding, where intermediate nodes forward linear
combinations of packets while preserving integrity~\cite{AgrawalBoneh2009HomMAC}, and subsequent work generalized homomorphic
authentication to broader classes of computations and more practical instantiations~\cite{GennaroWichs2013FHMA,CatalanoFiore2013PracticalHomMAC}.
\textsf{2G2T} can be viewed as an especially simple linearly-homomorphic authenticator specialized to MSM: the public merged-bases
vector $\vec{T}$ plays the role of an evaluation key, while the verifier's secret state binds the server to the \emph{entire}
MSM instance.

\paragraph{General verifiable computation.}
Finally, MSM outsourcing can be viewed as a special case of verifiable computation and delegation, where the goal is to verify
correctness with substantially less work than direct evaluation.
General-purpose frameworks such as GKR-style interactive proofs and non-interactive verifiable computing formalize this goal,
but are typically heavier than what is desirable for a ubiquitous cryptographic kernel~\cite{Goldwasser2008GKR,Gennaro2010NIVC}.

\section{Performance evaluation}\label{sec:2g2t-performance}

This section empirically evaluates the computational cost of \textsf{2G2T} and quantifies the
speedup obtained by verifying an outsourced MSM (via the \textsf{2G2T} check) rather than
computing the MSM locally.
All benchmarks were run on an Apple~M4 machine with 32\,GB of memory running macOS~Tahoe,
compiled with \texttt{rustc 1.93.0 (254b59607 2026-01-19)} and using \texttt{curve25519-dalek v4.1.3}.
Figure~\ref{fig:2g2t-speedup} summarizes our main results.

\paragraph{Implementation setting and curve choice.}
We implement a reference benchmark in Rust using the \texttt{curve25519-dalek} library over the
Ristretto255 group. Ristretto255 is a natural instantiation for our target setting:
it provides a prime-order group abstraction,
it is transparent (no pairings, no SRS), and it is commonly used in Bulletproofs-style deployments.
Throughout, we fix the public point $Q$ to the standard Ristretto basepoint so that the verifier's
fixed-base multiplication $sQ$ can leverage the library's precomputed basepoint table.\footnote{The
group-only part of verification is already negligible in our measurements; choosing $Q$ as the
basepoint simply reflects a realistic engineering choice for an always-fixed public $Q$.}

For each dimension $n=2^e$ (with $e\in\{10,\dots,18\}$), we sample $\vec{P}\in\mathbb{G}^n$ and
$\vec{x}\in\mathbb{F}_q^n$ at random, run the \textsf{2G2T} precomputation once to obtain $\vec{T}$,
and then measure:
(i) the time to compute $\MSM(\vec{P},\vec{x})$ under selected baselines, and
(ii) the time for the client-side \textsf{2G2T} verification
(i.e., compute $s=\langle \vec{x},\vec{\rho}\rangle$ and check $B \stackrel{?}{=} rA+sQ$).
We report per-iteration averages over many iterations (after warmup) and focus on \emph{ratios} (speedups),
which are relatively stable across machines.

\paragraph{Verification cost and ``lazy'' inner product.}
In \textsf{2G2T}, the verifier's dominant term is computing the scalar
\[
s = \langle \vec{x},\vec{\rho}\rangle = \sum_{i=1}^n x_i\rho_i \in \mathbb{F}_q,
\]
followed by a constant amount of group work ($rA$, $sQ$, one addition, one comparison).
To reduce modular-reduction overhead in the scalar field, we implement a \emph{lazy-reduction} inner product:
we accumulate products in a wide (512-bit) accumulator and reduce modulo $q$ only once per block of terms.
In the experiments of Figure~\ref{fig:2g2t-speedup}, we use a block size of $128$, which keeps the
wide accumulator safely within 512 bits while substantially reducing the number of reductions.
This optimization directly targets the verifier bottleneck and is faithful to how a
high-performance \textsf{2G2T} verifier would be engineered.

\paragraph{Baselines and why we include a ``na\"ive'' MSM.}
Figure~\ref{fig:2g2t-speedup} plots speedup curves against two MSM baselines:

\begin{itemize}
  \item \textbf{Optimized MSM baseline (library MSM).}
  We use \texttt{curve25519-dalek}'s optimized multi-scalar multiplication routine\\
  (\texttt{RistrettoPoint::vartime\_multiscalar\_mul}).
  This is the appropriate baseline for performance-oriented deployments and captures the best
  available software MSM in our chosen setting.

  \item \textbf{Na\"ive MSM baseline (sum of scalar multiplications).}
  We also measure the straightforward method
  \[
  \sum_{i=1}^n x_i P_i \quad\text{computed as}\quad
  \text{acc}\gets 0;\ \text{for }i:\ \text{acc}\gets \text{acc} + x_iP_i.
  \]
  While asymptotically suboptimal, this baseline is still meaningful in practice and is worth reporting:
  (i) it approximates behavior on memory-constrained devices where bucket-based MSM is unattractive due to
  RAM/cache pressure; (ii) it reflects implementations that prioritize simplicity/auditability or constant-time
  scalar multiplication primitives (where MSM optimizations are not available or are disabled);
  and (iii) it provides a clear reference point that isolates how much of the gain comes from \textsf{2G2T}
  verification versus advanced MSM engineering.
\end{itemize}

We emphasize that the \emph{optimized} MSM baseline is the primary one; the na\"ive MSM curve is included as an
additional, practically motivated reference.

\paragraph{Results and interpretation.}
Figure~\ref{fig:2g2t-speedup} reports the speedup defined as
\[
\mathsf{Speedup}(n)\;:=\;\frac{t_{\mathrm{MSM}}(n)}{t_{\mathrm{verify}}(n)},
\]
where $t_{\mathrm{MSM}}(n)$ is the time to compute $\MSM(\vec{P},\vec{x})$ under the chosen baseline and
$t_{\mathrm{verify}}(n)$ is the client-side \textsf{2G2T} verification time.

The key observation is that \textsf{2G2T} verification is extremely lightweight even for large $n$:
the group-only part of the check is essentially constant (two scalar multiplications with one fixed base,
one addition, one equality test), while the only $n$-dependent work is the inner product.
Consequently, the speedup curve increases with $n$ (as fixed verification overhead is amortized) and then
plateaus once verification is dominated by the linear-time inner product.
At the largest tested dimension ($n=2^{18}=262{,}144$), verification takes on the order of a few milliseconds
in our implementation, whereas recomputing the MSM requires hundreds of milliseconds under the optimized MSM
routine and several seconds under the na\"ive baseline.
This corresponds to speedups on the order of
\emph{hundreds of times} against the optimized MSM baseline and
\emph{thousands of times} against the na\"ive baseline at large $n$.

Overall, the measurements support the main performance claim of \textsf{2G2T}:
the verifier can check an outsourced MSM with (i) constant-size server response (two group elements),
(ii) essentially constant group work, and (iii) only a single streaming pass over $\vec{x}$ for the inner product,
yielding large end-to-end speedups in the large-$n$ regime relevant to aggregated/inner-product-based proof systems.

\begin{figure}
    \centering
    \includegraphics[width=0.9\linewidth]{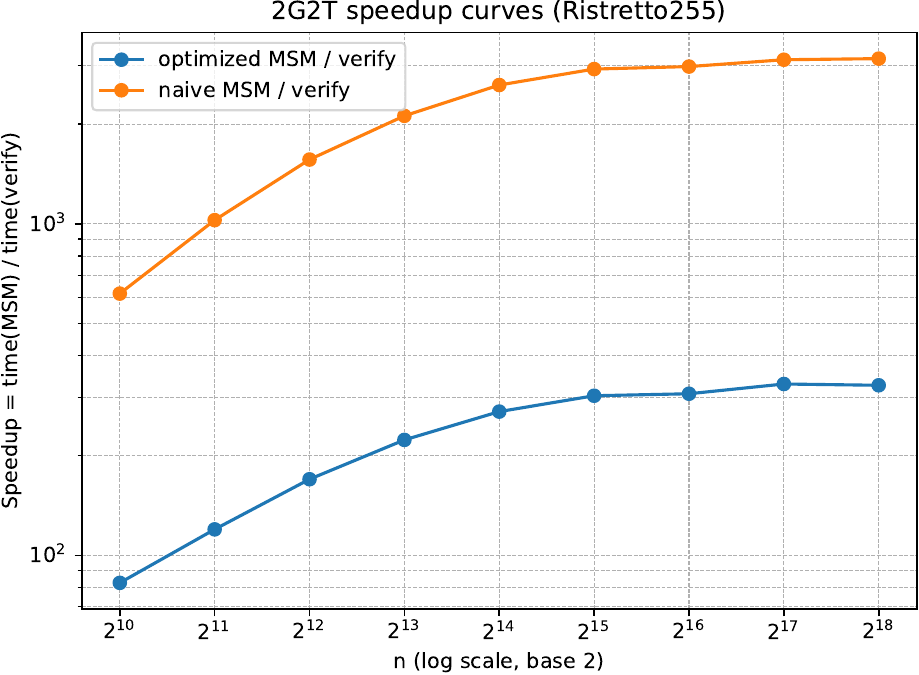}
    \caption{\textsf{2G2T} speedup on Ristretto255 (up to $\sim300\times$ vs.\ optimized MSM and $\sim3000\times$ vs.\ na\"ive MSM, for $n\le 2^{18}$).}
    \label{fig:2g2t-speedup}
\end{figure}

\paragraph{Comparison to Pinkas.}
For context, Pinkas~\cite{Pinkas2025VerifiableMSM} proposes a \emph{publicly verifiable} statistical check for MSMs in
prime-order groups. In his analysis, when scalars have bit-length $m$ (typically $m\approx \log_2 q$), verification reduces to
computing an $n$-wide MSM with \emph{shorter} $(\lambda+\log_2 m)$-bit exponents, yielding an expected speedup of roughly $\frac{m}{\lambda+\log_2 m}$.
Over Ristretto255, one may take $m\approx \log_2 q\approx 252$ and $\log_2 m\approx 8$, which predicts a speedup of about
$252/(64+8)\approx 3.5\times$ at $\lambda=64$.
In contrast, matching \textsf{2G2T}'s per-query soundness error $1/q$ corresponds to $\lambda\approx \log_2 q$ (about $252$),
in which case the predicted speedup essentially vanishes.
Pinkas' protocol also requires the server to return $m$ intermediate group elements (about $252$--$256$ for Ristretto255),
whereas \textsf{2G2T} returns only two group elements $(A,B)$.
These protocols therefore occupy different points in the design space: Pinkas provides public verifiability with
$\Theta(\log q)$ communication, while \textsf{2G2T} achieves constant-size responses and very fast designated-verifier
checks after a one-time keyed setup.

\section{Conclusion}
\label{sec:conclusion}

We introduced \textsf{2G2T}, a lightweight protocol for verifiably outsourcing generic-base multi-scalar multiplication
$\MSM(\vec{P},\vec{x})$ in transparent prime-order groups.
Our analysis shows that \textsf{2G2T} is perfectly complete and achieves \emph{full-group} statistical soundness:
even against an unbounded adversarial server, the probability of accepting an incorrect $A$ is at most $1/q$ per execution,
and at most $e/q$ over $e$ adaptive executions, requiring only that $\mathbb{G}$ has prime order $q$.
Empirically, a Rust implementation over Ristretto255 confirms that verification is extremely fast in the large-$n$ regime:
for $n$ up to $2^{18}$, verification is up to $\sim 300\times$ faster than locally computing an optimized MSM, while the
server-to-client response remains constant-size (two compressed group elements, 64 bytes on Ristretto255).
Overall, \textsf{2G2T} isolates a simple and practical primitive for integrity-preserving MSM acceleration with minimal
communication and verifier overhead.

\paragraph{Limitations and future work.}
\textsf{2G2T} is a designated-verifier protocol with a one-time keyed setup for fixed bases, and we do not aim to hide the
coefficients $\vec{x}$ from the server.
An important direction is extending \textsf{2G2T} to prover-side outsourcing, where $\vec{x}$ may be witness-dependent and must
remain hidden, while preserving verification that is far cheaper than a full MSM.
A second direction is exploring publicly verifiable variants, likely trading off larger proofs and/or stronger assumptions
for removing the verifier's secret state.

\clearpage

\bibliographystyle{splncs04}
\bibliography{citations}

\end{document}